\documentclass[conference]{IEEEtran}

\usepackage{cite}
\usepackage{amsmath}
\usepackage{amsfonts}
\usepackage{pifont}
\usepackage{amssymb}
\usepackage{amsthm}
\usepackage{tikz}
\usepackage{graphicx}
    \graphicspath{{../}}
    \DeclareGraphicsExtensions{.pdf}
\usepackage[caption=false,font=footnotesize]{subfig}
\usepackage{multirow}
\usepackage{url}
\usepackage{algorithmic}
\usepackage{enumerate}
\usepackage{makecell}
\theoremstyle{plain}

\newtheorem{theorem}{Theorem}
\newtheorem{lemma}[theorem]{Lemma}

\newtheorem{definition}[theorem]{Definition}

\theoremstyle{definition}

\begin{document}
\title{Achieving Arbitrary Locality and Availability in Binary Codes}

\author{\IEEEauthorblockN{Anyu~Wang and
        Zhifang~Zhang}

\IEEEauthorblockA{Key
Laboratory of Mathematics Mechanization, NCMIS\\
Academy of Mathematics and Systems Science, CAS, Beijing, China\\
Email: \{wanganyu, zfz\}@amss.ac.cn}
}
\maketitle
\thispagestyle{empty}

\begin{abstract}
The $i$th coordinate of an $(n,k)$ code is said to have locality $r$ and availability $t$ if there exist $t$ disjoint groups, each containing at most $r$ other coordinates that can together recover the value of the $i$th coordinate. This property is particularly useful for codes for distributed storage systems because it permits local repair and parallel accesses of hot data.
In this paper, for any positive integers $r$ and $t$, we construct a binary linear code of length $\binom{r+t}{t}$ which has locality $r$ and availability $t$ for all coordinates.
The information rate of this code attains $\frac{r}{r+t}$, which is always higher than that of the direct product code, the only known construction that can achieve arbitrary locality and availability.

\end{abstract}

\section{Introduction}\label{sec1}
Nowadays various forms of data redundancy are used in coding for distributed storage systems to insure data integrity and to improve performance efficiency. Among those the code with \emph{locality} $r$ has become an attractive subject since it was proposed independently by Gopalan et al. \cite{gopalan2012locality}, Oggier et al. \cite{oggier2011self}, and Papailiopoulos et al. \cite{papailiopoulos2012simple}.
More precisely, the $i$th coordinate of a code is said to have locality $r$ if the value at this coordinate can be recovered by accessing at most $r$ other coordinates. In other words, associating each coordinate with a storage node in a distributed storage system, an $(n,k)$ code with repair locality $r~(r \ll k)$ can greatly reduces the disk I/O complexity for node repair. Considering data reliability and storage efficiency, a lot of work studied upper bounds on the minimum distance and  information rate of such codes \cite{gopalan2012locality,forbes2013locality,papailiopoulos2012locally,cadambe2013upper,tamo2014family,wang2014integer}.
Codes attaining these upper bounds were constructed in \cite{tamo2013optimal,tamo2014family,huang2007pyramid}, and some have even found their way into practice \cite{huang2012erasure,sathiamoorthy2013xoring}. By now, a tight upper bound $\frac{r}{r+1}$ has been proven for the information rate, while a complete description for the tight upper bound on the minimum distance remains open.

Things become more complicated when people start considering locality $r$ in the case of multiple node erasures. First, locality $r$ that can tolerate up to $\delta-1$ erasures is realized in \cite{prakash2012optimal,silberstein2013optimal,song2014optimal} by using inner-error-correcting codes of length at most $r+\delta-1$ and minimum distance at least $\delta$. Then another way is developed in \cite{pamiesjuarez2013locally,wang2013repair} by providing $\delta-1$  disjoint local repair groups. A general framework for these works is built in \cite{wang2014repair}. Besides specific code structures, the methods of sequential local repair \cite{prakash2014codes} and cooperative local repair \cite{rawat2014cooperative} are also used for multiple erasures.

Recently, the property of $t$-\emph{availability} which is related to the structure of $\delta-1$ disjoint local repair groups (e.g. $t=\delta-1$ in \cite{wang2013repair}) is investigated further in \cite{rawat2014locality,tamo2014bounds}.
This property is particularly useful because it permits access of a coordinate from multiple ways in parallel, which is appealing in distributed storage systems with \emph{hot data}.

Specifically, the $i$th coordinate of a code is said to have locality $r$ and availability $t$ if there exist $t$ disjoint groups, each containing at most $r$ other coordinates that can together recover the value of the $i$th coordinate.
Codes that have locality $r$ and availability $t$ for information coordinates (e.g., all systematic
coordinates in a systematic linear code) are studied in \cite{rawat2014locality,wang2013repair}.
More precisely, the authors in \cite{wang2013repair} give an upper bound on the minimum distance for any $[n,k,d]_q$ linear codes that have locality $r$ and availability $t$ for information coordinates, i.e.,
\begin{equation*}
d \le n-k+2 - \lceil \frac{t(k-1)+1}{t(r-1)+1} \rceil.
\end{equation*}
They further prove the existence of codes attaining this upper bound when $n \ge k(r t+1)$.
In \cite{rawat2014locality}, an upper bound is derived for a special class of codes of which any repair group contains only $1$ parity symbol, and some explicit constructions attaining this bound are given there.

In this paper we focus on
codes that have locality $r$ and availability $t$ for {\it all coordinates}, for which the only known bounds are due to Tamo et al.\cite{tamo2014bounds}, i.e.
\begin{equation}\label{EqTamoRateBound}
\frac{k}{n} \le \prod_{i=1}^t \frac{1}{1+\frac{1}{ir}},
\end{equation}
and
\begin{equation*}
d \le n - \sum_{i=0}^t \lfloor \frac{k-1}{r^i} \rfloor.
\end{equation*}
These bounds are proven for all (linear or nonlinear) $(n,k)$ codes by using graph methods. There remains much work to be done in this field, such as discussing tightness of the upper bounds, constructing explicit codes which are optimal with respect to some upper bound, and etc. By far, the only known construction of codes achieving any given locality $r$ and availability $t$ is the direct product code (see, e.g., \cite{tamo2014bounds,wang2013repair}), while other constructions are given for special values of $r$ and $t$.
Table \ref{TableConstructions} lists almost all previous constructions of codes that have locality $r$ and availability $t$ for all coordinates.
\renewcommand\arraystretch{1.4}
\renewcommand\thefootnote{\fnsymbol{footnote}}
\begin{table}[htb]
\centering
\scriptsize
\begin{tabular}[b]{|c|c|c|}
\hline
 & $n,k,d$ & $r,t$ \\ \hline
\makecell{\cite{tamo2014bounds,wang2013repair}: \\ Direct product code}  & $n=(r+1)^t,k=r^t,d=2^t$ & $\forall r $, $\forall t$ \\ \hline
\makecell{\cite{macwilliams1977theory}: Simplex code}  & \makecell{$n=2^m-1$,\\$ k=m$, $d=2^{m-1}$} & \makecell{$r=2$, \\ $t=2^{m-1}-1$ }\\ \hline
\makecell{\cite[Example 1]{prakash2014codes}: \\ complete graph} & $n=\binom{r+2}{2},k=\binom{r+1}{2},d=3$ & $\forall r$, $t=2$ \\ \hline
\makecell{\cite[Construction 3]{tamo2014family}: \\ orthogonal partition }& $n,k,d=n-m+1$\footnotemark[2] & $\forall r$, $t=2$ \\ \hline
\makecell{\cite[Construction 4]{goparaju2014binary}: \\ tensor product matrix}& \makecell{$n=2^m-1$,\\$k=\frac{3}{7}n$, $d=4$} & $r=2,t=3$ \\ \hline
\makecell{Construction in\\ this paper} & \makecell{$n = \binom{r+t}{t},k=\binom{r+t-1}{t}$, \\ $d=t+1$} & $\forall r $, $\forall t$ \\
\hline
\end{tabular}
\caption{}
\label{TableConstructions}
\end{table}
\footnotetext[2]{In \cite{tamo2014family}, the specific values of $n,k$ and $m$ depend on the corresponding orthogonal partition and the encoding map, and the constraints are too complicated to be stated here.}


\subsection{Our Contribution}
For any positive integers $r$ and $t$, we construct a linear code of length $\binom{r+t}{t}$ which has locality $r$ and availability $t$ for all coordinates. Besides arbitrary locality and availability for all coordinates,
the following aspects make the code more desirable.
\begin{itemize}
\item[(1)]The code is over the binary field, which means efficient implementation in practice.
\item[(2)]Its information rate attains $\frac{r}{r+t}$ which is always higher than that of the direct product code.
Although no specific bound on the information rate is newly built in this work, through detailed comparisons with previous constructions and some related bounds, we believe our code has near optimal information rate when $t$ is not too large (say, $t<r$).
\end{itemize}

\subsection{Organization}
Section II introduces formal definitions of locality and availability, as well as some basic concepts about block design.
Section III presents the code construction and reveals its relation with the block design. Section IV gives comparisons with previous constructions and information rate bounds.
Section V concludes the paper.

\section{Locality  and Availability }
Let $\mathcal{C}$ be an $[n,k,d]_q$ linear code with generator matrix $G = (g_1,\dots,g_n)$, where $g_i$ is a $k$-dimensional column vector over $\mathbb{F}_q$ for $1 \le i \le n$.
Denote $[m]=\{1,2,\cdots,m\}$ for any positive integer $m$.
Then the locality $r$ and availability $t$ for linear codes are formally defined below.
\begin{definition}\label{DefAvailability}
The $i$th coordinate, $1 \le i \le n$, of an $[n,k,d]_q$ linear code $\mathcal{C}$ is said to have locality $r$ and availability $t$ if there exist $t$ disjoint subsets $R^{(i)}_1,\dots,R^{(i)}_t \subseteq [n]\setminus \{i\}$ such that for $1 \le j \le t$,
\begin{itemize}
\item[(1)]~$|R^{(i)}_j| \le r$, and
\item[(2)]~$g_i$ is an $\mathbb{F}_q$-linear combination of $\{g_l\}_{l \in R^{(i)}_j}$.
\end{itemize}
\end{definition}
In this paper, we prove locality $r$ and availability $t$ by verifying some equivalent conditions on the dual code. The details can be found in the proof of Theorem \ref{th4}. Then we also need some basic concepts of block design.
\begin{definition}
Let $X$ be a $v$-set (i.e. a set with $v$ elements), whose elements are called points. A $t$-$(v,k,\lambda)$ design is a collection of distinct $k$-subsets (called {\it blocks}) of $X$ with the property that any $t$-subset of $X$ is contained in exactly $\lambda$ blocks.
\end{definition}
Given a $t$-$(v,k,\lambda)$ design with $v$ points $P_1,...,P_v$ and $b$ blocks $B_1,...,B_b$, its $b\times v$ {\it incidence matrix} $A=(a_{ij})$ is defined by
$$a_{ij}=\left\{\begin{array}{ll}1,~&\mbox{if $P_j\in B_i$}\\0,~&\mbox{if $P_j\not\in B_i$}\end{array}\right.$$where $1\leq i\leq b$ and $1\leq j\leq v$.

\section{Code Construction}\label{secCons}
The code is constructed by defining its parity check matrix. For any positive integers $r$ and $t$, let $m=r+t$.
In the following, we define a matrix over $\mathbb{F}_2$, denoted as $H(m,t)$, containing $\binom{m}{t-1}$ rows and $\binom{m}{t}$ columns. Each row of $H(m,t)$ is associated with a $(t-1)$-subset of $[m]$ and each column of $H(m,t)$ is associated with a $t$-subset of $[m]$. Given the elements in $[m]$ ordered as $1\succ 2\succ\cdots\succ m$, we arrange the rows (also columns) of $H(m,t)$ in the lexical order of the associated subsets of $[m]$. More precisely, for any two subsets $E,F\subseteq [m]$ with the elements in each subset sorted in the order $\succ$, then $E$ is before $F$ if and only if for the first elements where $E$ and $F$ differ, say $a$ in $E$ and $b$ in $F$, it holds $a\succ b$. In this order, for $1\leq i\leq \binom{m}{t-1}$ and $1\leq j\leq \binom{m}{t}$, suppose the $i$th row is associated with the subset $E_i$ and the $j$th column is associated with the subset $F_j$, then the $(i,j)$th element $h_{ij}$ of $H(m,t)$ is defined as follows:
$$h_{ij}=\left\{\begin{array}{ll}1,~&\mbox{if $E_i\subseteq F_j$}\\0,~&\mbox{if $E_i\nsubseteq F_j$}\end{array}\right.$$

Let us give an example of $H(m,t)$. Suppose $t=3$ and $m=5$. The matrix $H(5,3)$ is given in Fig. \ref{f1}.
Actually, the matrix $H(m,t)$ is the parity check matrix of the code (denoted as $\mathcal{C}$) we construct in this section. Next we prove some properties of $H(m,t)$ to help understand the code $\mathcal{C}$.
{\renewcommand{\arraystretch}{1.1}
\begin{figure}[!tp]
$$\begin{array}{@{}r@{}c@{}c@{}c@{}c@{}c@{}c@{}c@{}c@{}c@{}c@{}l@{}}
    & 1 & 1 & 1 & 1  & 1 & 1 & 2 & 2 & 2 & 3\\
    & 2 & 2 & 2 & 3  & 3 & 4 & 3 & 3 & 4 & 4\\
    & 3 & 4 & 5 & 4  & 5 & 5 & 4 & 5 & 5 & 5\\\hline
    \left.\begin{array}
    {c}1,2\\1,3\\1,4\\1,5\\2,3\\2,4\\2,5\\3,4\\3,5\\4,5\end{array}\right(
                    & \begin{array}{c}1\\1\\0\\0\\1\\0\\0\\0\\0\\0 \end{array}
                    & \begin{array}{c}1\\0\\1\\0\\0\\1\\0\\0\\0\\0\end{array}
                    & \begin{array}{c}1\\0\\0\\1\\0\\0\\1\\0\\0\\0\end{array}
                    & \begin{array}{c}0\\1\\1\\0\\0\\0\\0\\1\\0\\0\end{array}
                    & \begin{array}{c}0\\1\\0\\1\\0\\0\\0\\0\\1\\0\end{array}
                    & \begin{array}{c}0\\0\\1\\1\\0\\0\\0\\0\\0\\1\end{array}
                    & \begin{array}{c}0\\0\\0\\0\\1\\1\\0\\1\\0\\0\end{array}
                    & \begin{array}{c}0\\0\\0\\0\\1\\0\\1\\0\\1\\0\end{array}
                    & \begin{array}{c}0\\0\\0\\0\\0\\1\\1\\0\\0\\1\end{array}
                    & \begin{array}{c}0\\0\\0\\0\\0\\0\\0\\1\\1\\1\end{array}
                          & \left)\begin{array}{c} \\ \\ \\  \\ \\ \\ \\ \\ \\ \\\end{array}\right.
  \end{array}$$
\caption{The matrix $H(5,3)$}\label{f1}
\end{figure}}

\begin{lemma}
  For $m>t>1$, the matrix $H(m,t)$ is of the block form
  \begin{equation}\label{eq1}
  H(m,t)=\begin{pmatrix}
  H(m-1,t-1)&{\bf 0}\\I_{\scriptscriptstyle\binom{m-1}{t-1}}&H(m-1,t)
  \end{pmatrix}\end{equation}
where $I_{\scriptscriptstyle\binom{m-1}{t-1}}$ is the unit matrix of size $\binom{m-1}{t-1}$ and ${\bf 0}$ is a zero matrix.
Particularly, for $m=t$, $H(m,t)=H(m,1)^\tau=(1,...,1)^\tau$ which is an all-one column vector of dimension $m$.
\end{lemma}
\begin{proof}
It is obvious that $H(m,1)=(1,...,1)$ and $H(m,m)=(1,...,1)^\tau$. For $m>t>1$, according to the order in which the rows (and columns) of $H(m,t)$ are arranged, the upper left block (i.e. the former $\binom{m-1}{t-2}$ rows and the former $\binom{m-1}{t-1}$ columns) corresponds to the subsets of $[m]$ containing $1$. Actually this block can be regarded as one defined over the set $\{2,3,...,m\}$ with columns corresponding to $(t-1)$-subsets and rows corresponding to $(t-2)$-subsets. Thus this block is the matrix $H(m-1,t-1)$. The upper right block of $H(m,t)$ is obviously ${\bf 0}$. Each row of the bottom left block corresponds a $(t-1)$-subset of $[m]\setminus\{1\}$ which is uniquely contained in a $t$-subset of $[m]$ containing $1$. Moreover, because the subsets are sorted in the lexical order, the bottom left block is the unit matrix of size $\binom{m-1}{t-1}$. Similar to the upper left block, it can see the bottom right block is $H(m-1,t)$.
\end{proof}

When $t > r+1$, the matrix $H(m=r+t,t)$ has more rows than columns. Actually, the following lemma states that the rows of $H(m,t)$ are linearly dependent, so some rows can be deleted when regarded as a parity check matrix.

\begin{lemma}\label{le2}
For the block decomposition of $H(m,t)$ as shown in (\ref{eq1}), each row in the upper block (i.e. $(H(m-1,t-1)~~{\bf 0})$) is a $\mathbb{F}_2$-linear combination of rows in the bottom block $(I_{\scriptscriptstyle\binom{m-1}{t-1}}~~H(m-1,t))$. Consequently, ${\rm rank}~H(m,t)=\binom{m-1}{t-1}$.
\end{lemma}
\begin{proof}
For any row (denoted as $h$) in the upper block, suppose  it is associated with a $(t-1)$-subset $\{1,a_1,...,a_{t-2}\}$ where $a_i\in[m]\setminus\{1\}$ for $1\leq i\leq t-2$. Denote $[m]=\{1,a_1,...,a_{t-2}\}\cup\{b_1,...,b_{m-t+1}\}$.  Then the row $h$ has $1$'s in the columns associated with
the subsets $\{1,a_1,...,a_{t-2},b_j\}$, $1\leq j\leq m-t+1$. As a result, the left part of $h$ is a sum of the left parts of the rows in the bottom block associated with the subsets $\{a_1,...,a_{t-2},b_j\}$, $1\leq j\leq m-t+1$. For simplicity, the collection of these rows is denoted by $R$.  We only need to show the sum (in $\mathbb{F}_2$) of the right parts of rows in $R$ is ${\bf 0}$.

Let us focus on the right part, it can see only the columns associated with the subsets $\{a_1,...,a_{t-2},b_i,b_j\}$, $1\leq i,j\leq m-t+1$, have $1$'s in the rows in $R$.
Furthermore, for each of these columns, say the column associated with $\{a_1,...,a_{t-2},b_i,b_j\}$, there are exactly two rows in $R$, i.e. the row associated with $\{a_1,...,a_{t-2},b_i\}$ and the row associated with $\{a_1,...,a_{t-2},b_j\}$, which have $1$ in that column. Consequently, the sum (in $\mathbb{F}_2$) of right parts of rows in $R$ is ${\bf 0}$. Since the rows in the bottom block of $H(m,t)$ are obviously linearly independent, it follows ${\rm rank}~H(m,t)=\binom{m-1}{t-1}$.
\end{proof}
It is easy to verify that ${\rm rank}~H(m,m)={\rm rank}~H(m,1)=1$ which coincides with the results in Lemma \ref{le2}. Then the parity check matrix of the code $\mathcal{C}$ can be taken as \begin{equation}\label{eq2}H=(I_{\scriptscriptstyle\binom{m-1}{t-1}}~H(m-1,t))\;.\end{equation} Therefore $\mathcal{C}$ is of length $\binom{m}{t}=\binom{r+t}{t}$ and information rate $1-\binom{m-1}{t-1}/\binom{m}{t}=\frac{m-t}{m}=\frac{r}{r+t}$. In the following, we continue to investigate the locality and availability of the code $\mathcal{C}$.
\begin{theorem}\label{th4}
The code $\mathcal{C}$ which has the parity check matrix $H(r+t,t)$ (or $H$ as defined in (\ref{eq2})) has locality $r$ and availability $t$ for all coordinates.
\end{theorem}
\begin{proof}
It is equivalent to prove that for each coordinate $i\in[\binom{r+t}{t}]$ there exist $t$ codewords in the dual code, say ${\bf c}_1,...,{\bf c}_t$, such that $|{\rm supp}~{\bf c}_j|=r+1$ and ${\rm supp}~{\bf c}_j~\cap~{\rm supp}~{\bf c}_l=\{i\}$ for $1\leq j\neq l\leq t$. Actually, we will see the rows of $H(r+t,r)$ are exactly these codewords.

First, for any row in $H(r+t,r)$, suppose it is associated with a $(t-1)$-subset $E\subseteq[r+t]$. Since $E$ is contained in $r+t-(t-1)=r+1$ $t$-subsets of $[r+t]$, this row has $r+1$ $1$'s. Namely, the support of each row is of size $r+1$.

Then, for each coordinate $i\in[\binom{r+t}{t}]$  which corresponds to a column of $H(r+t,t)$, suppose this column is associated with a $t$-subset $F_i\subseteq[r+t]$. Because $F_i$ contains $t$ $(t-1)$-subsets, there are $t$ rows which have $1$ in this column. We claim that excluding the coordinate $i$, supports of these $t$ rows are pairwise disjoint. Otherwise, assume there are two rows, say the $j$th row (denoted as $h_j$, associated with the subset $E_j$) and the $l$th row (denoted as $h_l$, associated with the subset $E_l$), such that $\{i,u\}\subseteq {\rm supp}~h_j~\cap~{\rm supp}~h_l$ for some $u\in[r+t]\setminus\{i\}$. It implies that $$E_j\subseteq F_i\cap F_u~{\rm and}~E_l\subseteq F_i\cap F_u\;.$$As a result, $E_j\cup E_l\subseteq F_i\cap F_u$. But the union (resp. intersection) of two different $(t-1)$-subsets (resp. $t$-subsets) is of size at least $t$ (resp. at most $t-1$), which leads a contradiction. Therefore, the $t$ rows are exactly the codewords ${\bf c}_1,...,{\bf c}_t$ we need to complete the proof.
\end{proof}

Finally we determine the minimum distance of the code $\mathcal{C}$ from its parity matrix $H$ given in (\ref{eq2}). Since each column in the right part of $H$, i.e. $H(m-1,t)$, has $t$ $1$'s and the left part of $H$ is a unit matrix, there exist $t+1$ columns in $H$ which are linearly dependent. Thus the minimum distance of  $\mathcal{C}$ is at most $t+1$. On the other hand, from the availability $t$ it can see any $t$ erasures are recoverable for $\mathcal{C}$, thus the minimum distance is at least $t+1$. Therefore the minimum distance of $\mathcal{C}$ is $t+1$.

\subsection{Relation with the block design}
Actually, the matrix $H(r+t,t)$ can be viewed as an incidence matrix of a $1$-$(\binom{r+t}{t},r+1,t)$ design. Moreover, suppose the blocks are $B_1,...,B_b$, where $b=\frac{t}{r+1}\cdot\binom{r+t}{t}=\binom{r+t}{t-1}$, then it  holds
\begin{equation}\label{con3}|B_i\cap B_j|\leq 1 \mbox{~for~} 1\leq i<j\leq b
\end{equation}
In other words, once we find a $1$-$(n,r+1,t)$ design with blocks satisfying the condition (\ref{con3}), it immediately  derives a linear code of length $n$ with locality $r$ and availability $t$ by taking its incidence matrix as the parity check matrix of the code. To make the resulting code has high information rate, the incidence matrix needs to have low rank comparing with its column size. However, it is difficult to find such $1$-designs, constructing those with incidence matrices of low rank is even harder. Our construction of the matrix $H(m,t)$ provides a good way to do such things. Besides, some constructions from geometry are also feasible. For example, the following matrix $H$ gives a  $1$-$(9,3,2)$ design satisfying the property (\ref{con3}), i.e.
$$H=\begin{pmatrix}
1&0&0&1&0&0&1&0&0\\0&1&0&0&1&0&0&1&0\\0&0&1&0&0&1&0&0&1\\1&1&1&0&0&0&0&0&0\\0&0&0&1&1&1&0&0&0\\0&0&0&0&0&0&1&1&1
\end{pmatrix}\;.$$
Then it induces a binary linear code with locality $r=2$ and availability $t=2$, but its information rate is $\frac{4}{9}$ which is less than $\frac{r}{r+t}$. In fact, this matrix $H$ corresponds to the direct product code construction. More details can be found in the next section.

\section{Comparisons with Other Constructions and Information Rate Bounds}
The information rate and minimum distance are two important parameters for evaluating an error-correcting code. It is well known that a tradeoff exists between these two parameters. Although the minimum distance of the code constructed in Section \ref{secCons} is $t+1$ which is the lowest for a locally repairable code with availability $t$, we will see it performs well in information rate through the following comparisons.

\subsection{Comparison with Other Constructions}
\subsubsection{The direct product code}
The direct product code, see e.g., \cite{wang2014repair,tamo2014bounds}, is another code that can achieve {\it arbitrary} locality and availability.
Specifically, the direct product of $t$ binary $(r+1,r)$ single-parity-check codes induces a code with locality $r$, availability $t$ and information rate $(\frac{r}{r+1})^t$.
The code we constructed in Section \ref{secCons} also achieves locality $r$ and availability $t$, but has information rate $\frac{r}{r+t}$.
Because $(1+\frac{1}{r})^t > 1+\frac{t}{r}$ for all $t>1$, it follows that $(\frac{r}{r+1})^t = 1/(1+\frac{1}{r})^t < 1/(1+\frac{t}{r})=\frac{r}{r+t}$.
Thus, when $t>1$, our code always has higher information rate than the direct product code with the same locality $r$ and availability $t$. Fig. \ref{FigRateCompasison1} and Fig. \ref{FigRateCompasison2} display the $t$-$\frac{k}n$ curves of the codes.

\begin{figure}[htb]
\includegraphics[width=0.48\textwidth]{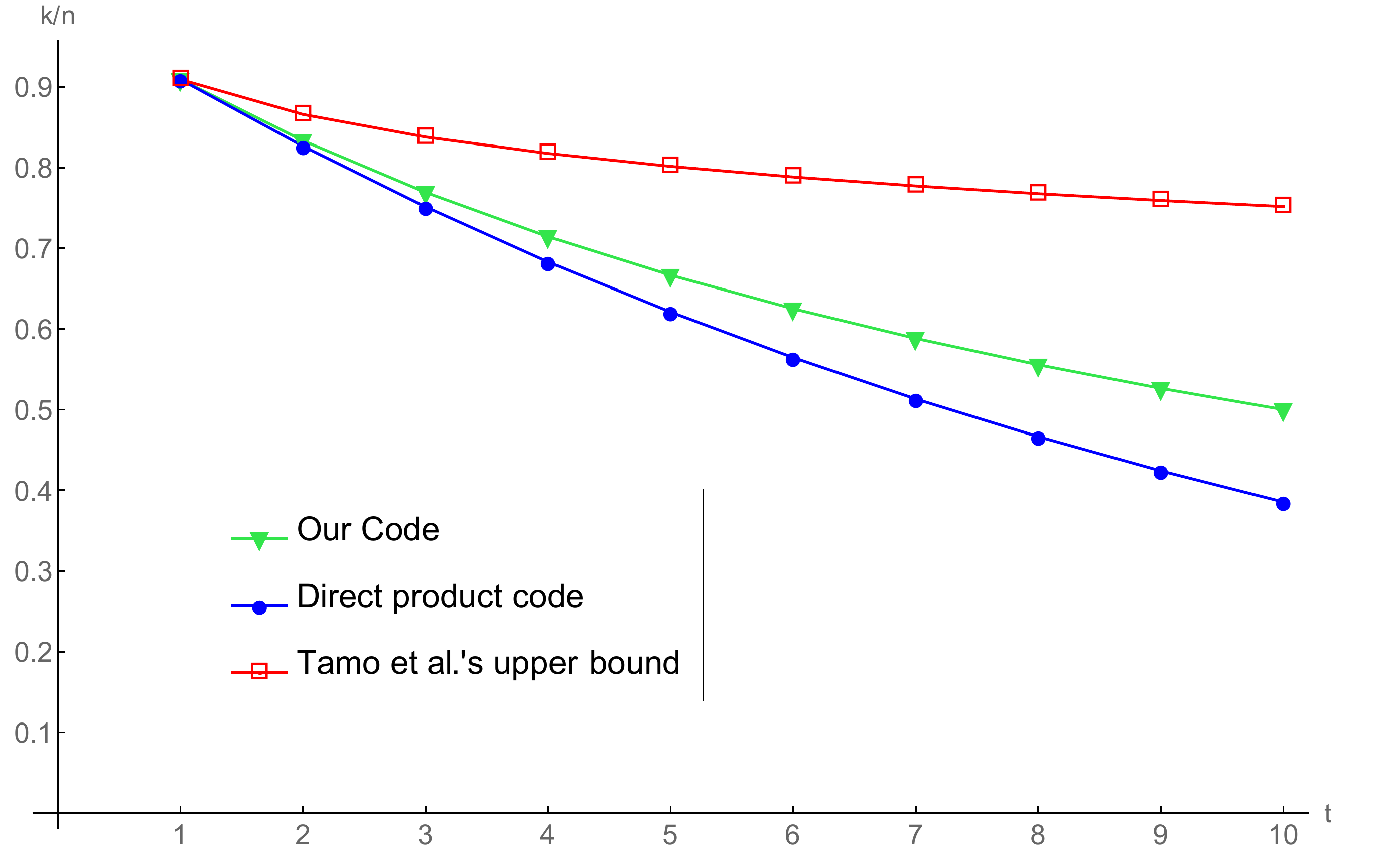}
\caption{Comparison of the information rate for $r=10$, $1 \le t \le 10$.}
\label{FigRateCompasison1}
\end{figure}

\begin{figure}[htb]
\includegraphics[width=0.48\textwidth]{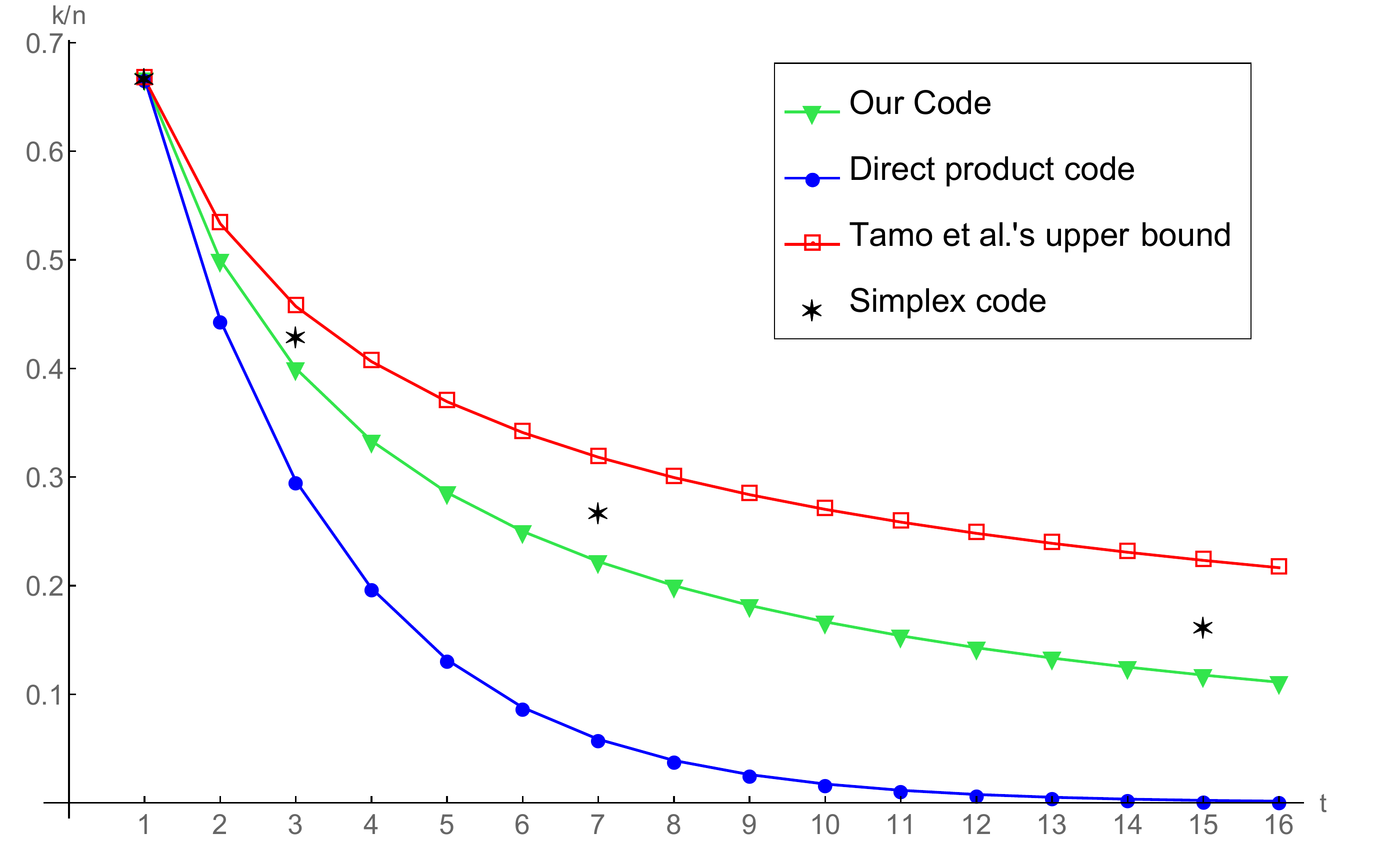}
\caption{Comparison of the information rate for $r=2$, $1 \le t \le 16$.}
\label{FigRateCompasison2}
\end{figure}

\subsubsection{Prakash et al's construction}
Recently, Prakash et al. \cite{prakash2014codes} presented a construction of locally $2$-reconstructible codes by using Tur{\'a}n graphs.
It was also shown in \cite{prakash2014codes} that the resulting codes have $2$-availability when using complete graphs instead of using Tur{\'a}n graphs.
Particularly, a complete graph with $v$ vertices induces a binary code with locality $r=v-1$ and availability $t=2$. This code has length $\frac{(r+1)(r+2)}{2}$ and information rate $\frac{r}{r+2}$.
In fact, after a permutation of the columns of the parity check matrix they constructed there, one can get exactly the matrix $H(r+2,2)$. That is, the code is equivalent to our construction for the very special case of $t=2$, while our work develops a general construction for arbitrary $t$.

\subsection{Comparison with Information Rate Bounds}
The only known bound on the information rate of locally repairable codes with availability $t$ for all coordinates is due to Tamo et al. \cite{tamo2014bounds} (see the bound (\ref{EqTamoRateBound}) stated in Section \ref{sec1}).
There does exist a gap between this upper bound and the rate we have achieved (i.e. $\frac{r}{r+t}$). But the following remarks help neutralize this difference.
\begin{itemize}
  \item[(1) ]To our knowledge, no codes attaining  the bound (\ref{EqTamoRateBound}) have been given (except for the special case of $t=1$). It was pointed out in \cite{tamo2014bounds} that  information rate of the direct product code is close to the bound (\ref{EqTamoRateBound}) for $t=2$. As we have seen, our code outperforms the direct product code in the information rate. Thus we get even closer to the upper bound (\ref{EqTamoRateBound}).

      Actually, in \cite{prakash2014codes} Prakash et al. derived an upper bound for the linear locally $2$-reconstructible code, i.e.,
\begin{equation}\label{EqPrakashRateBound}
\frac{k}{n} \le \frac{r}{r+2}.
\end{equation}
Since the locally repairable code with availability $2$ is a special class of locally $2$-reconstructible codes, the bound (\ref{EqPrakashRateBound}) also applies to the code we considered in this paper for the case $t=2$. On the other hand, our construction (also the construction given in \cite{prakash2014codes}) has proved tightness of the upper bound $\frac{r}{r+2}$.
As a comparison, when $t=2$ the upper bound (\ref{EqTamoRateBound}) exceeds $\frac{r}{r+2}$ by $\Omega(\frac1r)$.

  \item[(2) ]Our code is a binary one, while the bound (\ref{EqTamoRateBound}) is proved for all codes with locality $r$ and availability $t$. It is unsurprised that the field size compromises the information rate of a code sometimes.
  \item[(3) ]An upper bound on the information rate of codes with $(r,\delta)$ locality is given in \cite{song2014optimal}, i.e. $\frac{k}{n}\leq\frac{r}{r+\delta-1}$. The locality $(r,\delta)$ as introduced in \cite{prakash2012optimal} maintains the locality $r$ even in the case of $\delta-1$ erasures. It is easy to see that the requirement of locality $r$ and availability $t$ also guarantees the locality $r$ in the case of $t$ erasures. By letting $t=\delta-1$, our codes attains the upper bound for the codes with $(r,\delta)$ locality. Therefore, it is reasonable to believe that the information rate $\frac{r}{r+t}$ is near to the optimal for the codes with locality $r$ and availability $t$, especially for the case that $t$ is not too large (for example, $t<r$). As displayed in Fig. \ref{FigRateCompasison1}, we believe the curve of our code is closer to the optimal curve (which we have not obtained definitely) than the bound (\ref{EqTamoRateBound}).

  \item[(4) ]However, for large $t$, there do exist codes  which have information rate exceeding $\frac{r}{r+t}$.
      For example, the binary simplex code of length $n=2^m-1$ has locality $r=2$ and availability $t=2^{m-1}-1$. Its information rate is $\frac{m}{2^m-1}$ greater than $\frac{r}{r+t} = \frac{2}{2^{m-1}+1} $ for all $m \ge 3$. A comparison between all these codes and bounds is shown in Fig. \ref{FigRateCompasison2}. Therefore, to characterize the optimal information rate for codes with locality and availability, we still have a long way to go.
\end{itemize}

\section{Conclusions}
In this paper we construct a binary linear code with arbitrary locality $r$ and availability $t$.
The code can always have higher information rate than the direct product code which is the only known construction with the same property. Besides, it attains the optimal information rate at $t=2$.
This construction reveals a connection with special block designs, which may help to get more results on the codes with  locality and availability.


\begin{thebibliography}{100}

\bibitem{cadambe2013upper}
V.~Cadambe and A.~Mazumdar,
\newblock ``An upper bound on the size of locally recoverable codes,"
\newblock {\it IEEE Int. Symp.  Netw. Coding (NetCod),} Calgary, 2013, pp. 1--5.

\bibitem{forbes2013locality}
M.~Forbes and S.~Yekhanin,
\newblock ``On the locality of codeword symbols in non-linear codes,"
\newblock {\it arXiv preprint arXiv:1303.3921}, 2013.

\bibitem{gopalan2012locality}
P.~Gopalan, C.~Huang, H.~Simitci, and S.~Yekhanin,
\newblock ``On the locality of codeword symbols,"
\newblock {\it IEEE Trans. on Inform. Theory}, vol. 58, pp. 6925--6934, Nov. 2012.

\bibitem{goparaju2014binary}
S.~Goparaju and R.~Calderbank,
\newblock ``Binary cyclic codes that are locally repairable,"
\newblock in {\it Proc. IEEE Int. Symp. Inf. Theory (ISIT),} Honolulu, 2014, pp. 676--680.

\bibitem{huang2007pyramid}
C.~Huang, M.~Chen, and J.~Li,
\newblock ``Pyramid codes: Flexible schemes to trade space for access efficiency in reliable data storage systems,"
\newblock in {\it Proc. 6th IEEE Int. Symp. Netw. Comput. Appl.,} Cambridge, 2007, pp. 79¨C86.

\bibitem{huang2012erasure}
C.~Huang, H.~Simitci, Y.~Xu, A.~Ogus, B.~Calder, P.~Gopalan, J.~Li, and S.~Yekhanin,
\newblock ``Erasure coding in Windows Azure Storage,"
\newblock presented at the USENIX Annu. Tech. Conf., Boston, MA, 2012.

\bibitem{macwilliams1977theory}
F.~MacWilliams and N.~Sloane,
\newblock {\it The Theory of Error-Correcting Codes},
\newblock North-Holland, Amsterdam, The Netherlands, 1977.

\bibitem{oggier2011self}
F.~Oggier and A.~Datta,
\newblock ``Self-repairing homomorphic codes for distributed storage systems,"
\newblock in {\it Proc. IEEE Infocom,} Shanghai, 2011, pp. 1215--1223.

\bibitem{pamiesjuarez2013locally}
L.~Pamies-Juarez, H.~D.~L. Hollmann, and F.~Oggier,
\newblock ``Locally repairable codes with multiple repair alternatives,"
\newblock in {\it Proc. IEEE Int. Symp. Inf. Theory (ISIT),} Istanbul, 2013, pp. 892--896.

\bibitem{papailiopoulos2012locally}
D.~S. Papailiopoulos and A.~G. Dimakis,
\newblock ``Locally repairable codes,"
\newblock in {\it Proc. IEEE Int. Symp. Inf. Theory (ISIT),} Cambridge, 2012, pp. 2771--2775.

\bibitem{papailiopoulos2012simple}
D.~S. Papailiopoulos, J.~Luo, A.~G. Dimakis, C.~Huang, , and J.~Li,
\newblock ``Simple regenerating codes: network coding for cloud storage,"
\newblock in {\it Proc. IEEE Infocom,} Orlando, 2012, pp. 2801--2805.

\bibitem{prakash2012optimal}
N.~Prakash, G.~M. Kamath, V.~Lalitha, and P.~V. Kumar,
\newblock ``Optimal linear codes with a local-error-correction property,"
\newblock in {\it Proc. IEEE Int. Symp. Inf. Theory (ISIT),} Cambridge, 2012, pp. 2776--2780.

\bibitem{prakash2014codes}
N.~Prakash, V.~Lalitha, and P.~Kumar.
\newblock ``Codes with locality for two erasures,"
\newblock in {\it Proc. IEEE Int. Symp. Inf. Theory (ISIT),} Honolulu, 2014, pp. 1962--1966.

\bibitem{rawat2014locality}
A.~S. Rawat, D.~S. Papailiopoulos, A.~G. Dimakis, and S.~Vishwanath,
\newblock ``Locality and availability in distributed storage,"
\newblock in {\it Proc. IEEE Int. Symp. Inf. Theory (ISIT),} Honolulu, 2014, pp. 681--685.

\bibitem{rawat2014cooperative}
A.~S. Rawat, A. Mazumdar, and S. Vishwanath,
\newblock ``On cooperative local repair in distributed storage,"
\newblock in {\it 48th Annual Conference on Information Sciences and Systems (CISS),} Princeton, 2014, pp. 1--5.

\bibitem{sathiamoorthy2013xoring}
M.~Sathiamoorthy, M.~Asteris, D.~Papailiopoulos, A.~G. Dimakis, R.~Vadali,
  S.~Chen, and D.~Borthakur,
\newblock ``Xoring elephants: Novel erasure codes for big data,"
\newblock {\it Proceedings of the VLDB Endowment (to appear)}, 2013.

\bibitem{silberstein2013optimal}
N.~Silberstein, A.~S. Rawat, O.~O. Koyluoglu, and S.~Vishwanath,
\newblock ``Optimal locally repairable codes via rank-metric codes,"
\newblock in {\it Proc. IEEE Int. Symp. Inf. Theory (ISIT),} Istanbul, 2013, pp. 1819--1823.

\bibitem{song2014optimal}
W.~Song, S.~Dau, C.~Yuen, and T.~Li,
\newblock ``Optimal locally repairable linear codes,"
\newblock {\it IEEE J. Sel. Areas Commun.,} vol. 32, pp. 6925--6934, May 2014.

\bibitem{tamo2014bounds}
I.~Tamo and A.~Barg,
\newblock ``Bounds on locally recoverable codes with multiple recovering sets,"
\newblock in {\it Proc. IEEE Int. Symp. Inf. Theory (ISIT),} Honolulu, 2014, pp. 691--695.

\bibitem{tamo2013optimal}
I.~Tamo, D.~S. Papailiopoulos, and A.~G. Dimakis,
\newblock ``Optimal locally repairable codes and connections to matroid theory,"
\newblock in {\it Proc. IEEE Int. Symp. Inf. Theory (ISIT),} Istanbul, 2013, pp. 1814--1818.

\bibitem{tamo2014family}
I.~Tamo and A.~Barg,
\newblock ``A family of optimal locally recoverable codes,"
\newblock {\it IEEE Trans. on Inform. Theory}, vol. 60, pp. 4661--4676, Aug. 2014.

\bibitem{wang2013repair}
A.~Wang and Z.~Zhang,
\newblock ``Repair locality with multiple erasure tolerance,"
\newblock {\it IEEE Trans. on Inform. Theory}, vol. 60, pp. 6979--6987, Nov. 2014.

\bibitem{wang2014repair}
A.~Wang and Z.~Zhang,
\newblock ``Repair locality from a combinatorial perspective,"
\newblock in {\it Proc. IEEE Int. Symp. Inf. Theory (ISIT),} Honolulu, 2014, pp. 1972--1976.

\bibitem{wang2014integer}
A.~Wang and Z.~Zhang,
\newblock ``An integer progamming based bound for locally repaiable codes,"
\newblock in {\it arXiv preprint arXiv:1409.0952}, 2014.




\end{thebibliography}
\end{document}